\documentclass{scrartcl}

\newcommand{\tags}{\usetag{arxiv,full}}

\usepackage{tagging}
\tags

\tagged{icalp}{
	\usepackage[T1]{fontenc}
}
\usepackage[margin=1in]{geometry}

\usepackage{hyperref}
\usepackage{cite}
\usepackage{algorithm}
\usepackage{algpseudocode}
\tagged{twocol}{\algrenewcommand\algorithmicindent{1.0em}}

\usepackage{amsmath}
\usepackage{amssymb}
\tagged{arxiv}{
	\usepackage{amsthm}
}
\usepackage[capitalise]{cleveref}

\usepackage{csquotes}
\usepackage{enumitem}
\usepackage{footnote}
\usepackage{booktabs}
\usepackage{xcolor}

\tagged{arxiv}{
	\newtheorem{theorem}{Theorem}[section]%
	\newtheorem{lemma}[theorem]{Lemma}
	
	\newtheorem{definition}[theorem]{Definition}
}

\tagged{icalp}{
	\crefname{@theorem}{Theorem}{Theorems}
	\crefname{Definition}{Definition}{Definitions}
}

\newlist{inenum}{enumerate*}{1}
\setlist[inenum]{label=(\roman*)}

\makesavenoteenv{tabular}
\makesavenoteenv{table}

\providecommand{\diedge}[1]{(#1)}

\makeatletter
\def\old@comma{,}
\catcode`\,=13
\def,{%
	\ifmmode%
	\old@comma\discretionary{}{}{}%
	\else%
	\old@comma%
	\fi%
}
\makeatother

\newcommand{\fail}{\textbf{Fail}}
\newcommand{\junk}[1]{}

\algtext*{EndWhile}
\algtext*{EndIf}
\algtext*{EndFor}
\algtext*{EndProcedure} 

\title{Sampling Arbitrary Subgraphs Exactly Uniformly in Sublinear Time}
\author{
	Hendrik Fichtenberger
	\thanks{Department of Computer Science, TU Dortmund, Germany. Email: 	\url{hendrik.fichtenberger@tu-dortmund.de}. \href{https://orcid.org/0000-0003-3246-5323}{ORCID~iD: 0000-0003-3246-5323}}
	\and 
	Mingze Gao\thanks{Department of Computer Science, University of Sheffield, UK. Email: \url{noahgao0015@gmail.com}.}
	\and
	Pan Peng
	\thanks{Department of Computer Science, University of Sheffield, UK. Email: \url{p.peng@sheffield.ac.uk}. \href{https://orcid.org/0000-0003-2700-5699}{ORCID~iD: 0000-0003-2700-5699}.}
	}

\date{}

\begin{document}
	\maketitle

\begin{abstract}
	We present a simple sublinear-time algorithm for sampling an arbitrary subgraph $H$ \emph{exactly uniformly} from a graph $G$ with $m$ edges, to which the algorithm has access by performing the following types of queries: (1) degree queries, (2) neighbor queries, (3) pair queries and (4) edge sampling queries. The query complexity and running time of our algorithm are $\tilde{O}(\min\{m, \frac{m^{\rho(H)}}{\# H}\})$ and $\tilde{O}(\frac{m^{\rho(H)}}{\# H})$, respectively, where $\rho(H)$ is the fractional edge-cover of $H$ and $\# H$ is the number of copies of $H$ in $G$. For any clique on $r$ vertices, i.e., $H=K_r$, our algorithm is almost optimal as any algorithm that samples an $H$ from any distribution that has $\Omega(1)$ total probability mass on the set of all copies of $H$ must perform $\Omega(\min\{m, \frac{m^{\rho(H)}}{\# H\cdot (cr)^r}\})$ queries.
	
	Together with the query and time complexities of the $(1\pm \varepsilon)$-approximation algorithm for the number of subgraphs $H$ by Assadi, Kapralov and Khanna~\cite{assadi2018simple} and the lower bound by Eden and Rosenbaum \cite{total-lower-bound} for approximately counting cliques, our results suggest that in our query model, approximately counting cliques is ``equivalent to'' exactly uniformly sampling cliques, in the sense that the query and time complexities of exactly uniform sampling and randomized approximate counting are within polylogarithmic factor of each other. This stands in interesting contrast to an analogous relation between approximate counting and almost uniformly sampling for self-reducible problems in the polynomial-time regime by Jerrum, Valiant and Vazirani~\cite{jerrum1986random}.
\end{abstract}

\section{Introduction}

``\emph{Given a huge real graph, how can we derive a representative sample?}'' is a first question asked by Leskovec and Faloutsos in their seminal work on graph mining~\cite{leskovec2006sampling}, which is motivated by the practical concern that most classical graph algorithms are too expensive for massive graphs (with millions or billions of vertices), and  graph sampling seems essential for lifting the dilemma.

In this paper, we study the question of how to sample a subgraph $H$ uniformly at random from the set of all subgraphs that are isomorphic to $H$ contained in a large graph $G$ in \emph{sublinear time}, where the algorithm is given query access to the graph $G$. That is, the algorithm only probes a small portion of the graph while still returning a sample with provable performance guarantee. 
Such a question is relevant for statistical reasons: we might need a few representative and unbiased motifs from a large network \cite{triangle_counting_app_3}, or edge-colored subgraphs in a structured database~\cite{atserias2008size}, in a limited time. A subroutine for extracting a uniform sample of $H$ is also useful in streaming (e.g., \cite{ahmed2017sampling}), parallel and distributed computing (e.g., \cite{feng2017can}) and other randomized graph algorithms (e.g., \cite{hu2013survey}).

Currently, our understanding of the above question is still rather limited. Kaufman, Krivelevich and Ron gave the first algorithm for sampling an edge almost uniformly at random \cite{edge_sampling_begin}. Eden and Rosenbaum gave a simpler and faster algorithm~\cite{edge_sampling_1}. Both works considered the \emph{general graph model}, where an algorithm is allowed to perform the following queries, where each query will be answered in constant time: 
\begin{description}
	\item[uniform vertex query] the algorithm can sample a vertex uniformly at random;
	\item[degree query] for any vertex $v$, the algorithm can query its degree $d_v$;
	\item[neighbor query] for any vertex $v$ and index $i\leq d_v$, the algorithm can query the $i$-th neighbor of $v$;
	\item[pair query] for any two vertices $u,v$, the algorithm can query if there is an edge between $u,v$. 
\end{description}

In \cite{edge_sampling_1}, Eden and Rosenbaum gave an algorithm that takes as input a graph with $n$ vertices and $m$ edges (where $m$ is unknown to the algorithm), uses $\tilde{O}(n/\sqrt{m})$ queries\footnote{Throughout the paper, we use $\tilde{O}(\cdot)$ to suppress any dependencies on the parameter $\varepsilon$, the size of the corresponding subgraph $H$ and $\log(n)$-terms.} in expectation and returns an edge $e$ that is sampled with probability $(1\pm \varepsilon)/m$ (i.e., almost uniformly at random). This is almost optimal in the sense that any algorithm that samples an edge from an almost-uniform distribution requires $\Omega(n/\sqrt{m})$ queries. In their sublinear-time algorithm for approximately counting the number cliques \cite{eden2018approximating} (see below), Eden, Ron and Seshadhri use a procedure to sample cliques incident to a suitable vertex subset $S$ almost uniformly at random. However, for an arbitrary subgraph $H$, it is still unclear how to obtain an almost uniform sample in sublinear time.  

\subparagraph{Approximate counting in sublinear-time} In contrast to sampling subgraphs (almost) uniformly at random, the very related line of research on approximate counting the number of subgraphs in sublinear time has made some remarkable progress in the past few years. Feige gave a $(2+\varepsilon)$-approximation algorithm with $\tilde{O}(n/\sqrt{m})$ queries for the average degree, which is equivalent to estimating the number of edges, of a graph in the model that only uses vertex sampling and degree queries \cite{edge_counting_app_1}. He also showed that any $(2-o(1))$-approximation for the average degree using only vertex and degree queries requires 
$\Omega(n)$ queries. Goldreich and Ron then gave a $(1+\varepsilon)$-approximation algorithm with $\tilde{O}(n/\sqrt{m})$ queries for the average degree in the model that allows vertex sampling, degree and neighbor queries~\cite{edge_counting_0}. 

Eden et al. recently gave the first sublinear-time algorithm for $(1\pm \varepsilon)$-approximating the number of triangles~\cite{triangle_counting_1}. Later, Eden, Ron and Seshadhri generalized it to $(1\pm \varepsilon)$-approximating the number of $r$-cliques $K_r$ \cite{eden2018approximating} in the general graph model that allows vertex sampling, degree, neighbor and vertex-pair queries. The query complexity and running time of their algorithms for $r$-clique $K_r$ counting are $\tilde{O}(\frac{n}{(\# {K_r})^{1/3}} +\min\{m,\frac{m^{r/2}}{\# {K_r}}\})$ and $\tilde{O}(\frac{n}{(\# {K_r})^{1/3}} +\frac{m^{r/2}}{\# {K_r}})$ respectively, for any $r\geq 3$, where $\# {K_r}$ is the number of copies of $K_r$ in $G$. Furthermore, in both works it was proved that the query complexities of the respective algorithms are optimal up to polylogarithmic dependencies on $n, \epsilon$ and $r$.

Later, Assadi et al.~\cite{assadi2018simple} gave a sublinear-time algorithm for $(1\pm \varepsilon)$-approximating the number of copies of an arbitrary subgraph $H$ in the \emph{augmented general graph model}~\cite{AliSub17}. That is, besides the aforementioned vertex sampling, degree, neighbor and pair queries, the algorithm is allowed to perform the following type of queries:
\begin{description}
	\item[edge sampling query] the algorithm can sample an edge uniformly at random.
\end{description}

The algorithm in \cite{assadi2018simple} uses $\tilde{O}(\min\{m, \frac{m^{\rho(H)}}{\# H}\})$ queries and $\tilde{O}(\frac{m^{\rho(H)}}{\# H})$ time, where $\rho(H)$ is the fractional edge-cover of $H$ and $\# H$ is the number of copies of $H$ in $G$. For the special case $H=K_r$, their algorithm performs $\tilde{O}(\min\{m,\frac{m^{r/2}}{\# {K_r}}\})$ queries and runs in $\tilde{O}(\frac{m^{r/2}}{\# {K_r}})$ time, which do not have the additive term $\frac{n}{(\# K_r)^{1/3}}$ in the query complexity and running time of the algorithms in \cite{triangle_counting_1,eden2018approximating}.
Eden and Rosenbaum provided simple proofs that most of the aforementioned results are nearly optimal in terms of their query complexities by reducing from communication complexity problems~\cite{total-lower-bound}. Further investigation of sampling an edge and estimating subgraphs in low arboricity graphs~\cite{eden_et_al:LIPIcs:2019:10628,ERS20} and approximately counting stars~\cite{AliSub17} has also been performed.

\subparagraph{Relation of approximate counting and almost uniform sampling} 
One of our original motivations is to investigate the relation of approximate counting and almost uniform sampling in the sublinear-time regime. That is, we are interested in the question whether \emph{in the sublinear-time regime, is almost uniform sampling ``computationally comparable'' to approximate counting, or is it strictly harder or easier, in terms of the query and/or time complexities for solving these two problems?} Indeed, in the polynomial-time regime, Jerrum, Valiant and Vazirani showed that for self-reducible problems (e.g., counting the number of perfect matchings of a graph), approximating counting is ``equivalent to'' almost uniform sampling~\cite{jerrum1986random}, in the sense that the time complexities of almost uniform sampling and randomized approximate counting are within polynomial factor of each other. Such a result has been instrumental for the development of the area of approximate counting (e.g., \cite{sinclair1989approximate}). It is natural to ask if similar relations between approximate counting and sampling hold in the sublinear-time regime. In \cite{eden_et_al:LIPIcs:2019:10628}, the authors mentioned that in the general graph model, the query complexities of approximate counting and almost uniformly sampling \emph{edges} are the same (up to $\log n, 1/\varepsilon$ dependencies), while there exist constant-arboricity graphs from which sampling \emph{triangles} almost uniformly requires $\Omega(n^{1/4})$ queries while approximately counting triangles can be done with $\tilde{O}(1)$ queries.

\subsection{Our Results}
In this paper, we consider the problem of (almost) uniformly sampling a subgraph in the augmented general graph model. As mentioned above, this model has been studied in \cite{AliSub17,assadi2018simple}, in which the authors find that ``allowing edge-sample queries results in considerably simpler and more general algorithms for
subgraph counting and is hence worth studying on its own''. On the other hand, allowing edge sampling queries is also natural in models where neighbor queries are allowed, e.g., in the well-studied bounded-degree model and the general model: most graph representations that allow efficient neighbor queries (e.g., GEXF, GML or GraphML) store edges in linear data structures, which often allows efficient (nearly) uniformly sampling of edges. We refer to \cite{assadi2018simple} for a deeper discussion on allowing edge sampling queries from both theoretical and practical perspectives.

We prove the following upper bound on sampling subgraphs (exactly) uniformly at random and provide a corresponding algorithm in \cref{sec:upperbound}.
\begin{theorem}\label{thm:main}
	Let $H$ be an arbitrary subgraph of constant size. There exists an algorithm in the augmented general graph model that given query access to the input graph $G=(V,E)$ and the number of edges $m$ in $G$ and uses $\tilde{O}(\min\{m, \frac{m^{\rho(H)}}{\# H}\})$ queries in expectation, and with probability at least $2/3$, returns a copy of $H$, if $\# H>0$. Each returned $H$ is sampled according to the uniform distribution over all copies of $H$ in $G$. The expected running time of the algorithm is $\tilde{O}(\frac{m^{\rho(H)}}{\# H})$.
\end{theorem}

We stress that our sampler is an exactly uniform sampler, i.e., the returned $H$ is sampled from the uniform distribution, while to the best of our knowledge, the previous sublinear-time subgraph sampling algorithms are only \emph{almost} uniform samplers. That is, they return an edge or a clique that is sampled from a distribution that is \emph{close} to the corresponding uniform distribution. Indeed, it has been cast as an open question if it is possible to sample an edge exactly uniformly at random in the general graph model in \cite{edge-sampling}. Furthermore, we remark that our algorithm actually does not perform any uniform vertex sampling query, a feature that might be preferable in practice. 

Our algorithm is based on one idea from \cite{assadi2018simple} (see also \cite{atserias2008size}) that uses the fractional edge cover to partition a subgraph $H$ into stars and odd cycles (i.e., Lemma \ref{decomposition-lemma}). The authors of \cite{assadi2018simple} also provided a scheme called \emph{subgraph-sampler trees} for recursively sampling stars and odd cycles that compose $H$, while the resulting distribution is not (almost) uniform distribution. Instead, we show that one can sample stars and odd cycles in parallel (or, more precisely, sequentially but independently of each other) and check whether they form a copy of $H$.

To complement our algorithmic result, we give a lower bound on the query complexity for sampling a clique in sublinear time by using a simple reduction from \cite{total-lower-bound}. We show the following theorem and present its proof in \cref{sec:lowerbound}.  
\begin{theorem}\label{thm:lowerbound}
	Let $r\geq 3$ be an integer. Suppose $\mathcal{A}$ is an algorithm in the augmented general graph model that for any graph $G=(V,E)$ on $n$ vertices and $m$ edges returns an arbitrary $r$-clique $K_r$, if one exists; furthermore, each returned clique $K_r$ is sampled according to a distribution $\mathcal{D}$, such that the total probability mass of $\mathcal{D}$ on the set of all copies of $K_r$ is $\Omega(1)$. Then $\mathcal{A}$ requires $\Omega(\min\{m, \frac{m^{r/2}}{\# K_r\cdot (cr)^r}\})$  queries, for some absolute constant $c>0$.
\end{theorem}

Note that the above theorem gives a lower bound for sampling $K_r$ from almost every non-trivial distribution $\mathcal{D}$. In particular, it holds if $\# K_r>0$ and $\mathcal{D}$ is a distribution that is only supported on the set of all copies of $K_r$, e.g., the (almost) uniform distribution on these copies. Together with the query and time complexities of the $(1\pm \varepsilon)$-approximation algorithm for the number of subgraphs $H$ by Assadi, Kapralov and Khanna~\cite{assadi2018simple} and the lower bound by Eden and Rosenbaum  \cite{total-lower-bound} for approximately counting cliques, our \cref{thm:main,thm:lowerbound} imply that in the augmented general graph model, \emph{approximately} counting the number of cliques is equivalent to \emph{exactly} sampling cliques in the sense that the query and time complexities of them are within a polylogarithmic factor of each other.

\subparagraph{Future Work} Considering real-world applications, it would be interesting to relax the guarantees of the queries available to the algorithm. In particular, one may not be able to sample edges \emph{exactly} uniformly at random, but only \emph{approximately} uniformly. For example, there exist works that consider weaker query models in which they sample vertices or edges almost uniformly at random by performing random walks from some fixed vertex (see, e.g., \cite{Ben-Hamou2018,Chiericetti2016}). Implementing these changes in the model would result in a weaker guarantee for the distribution of sampled subgraphs in \cref{thm:main} but would be potentially more practical.

\section{Preliminaries}
Let $G=(V,E)$ be a simple graph with $|V|=n$ vertices and $|E|=m$ edges. For a vertex $v \in V$, we denote by $d_v$ the degree of the vertex, by $\Gamma_v$ the set of all the neighbors of $v$, and by $E_v$ the set of edges incident to $v$. We fix a total order on vertices denoted by $\prec$ as follows:
\begin{definition}
	For any two vertices $u$ and $v$, we say that $u \prec v$ if $d_u < d_v$ or $d_u = d_v$ and $u$ appears before $v$ in the lexicographic order.  
\end{definition}

For any two vertices, we denote by $\Gamma_{uv}$ the set of the shared neighbors of $u$ and $v$ that are larger than $u$ with respect to ``$\prec$'', i.e., $\Gamma_{uv} = \{ w \mid w \in \Gamma_u \cap \Gamma_v \wedge u \prec w \}$. Sometimes, we view our graph $G=(V,E)$ as a directed graph $(V,\vec{E})$ by treating each undirected edge $e=\{u,v\}\in E$ as two directed edges $\vec{e}_1=(u,v)$ and $\vec{e}_2=(v,u)$. The following was proven in \cite{triangle_counting_1}.

\begin{lemma}[\cite{triangle_counting_1}]
	\label{lemma: tu-upper-bound}
	For any vertex $v$, the number of neighbors $w$ of $v$ such that $v\prec w$ is at most $\sqrt{2 m}$.
\end{lemma}

Given a graph $H$, we say that a subgraph $H'$ of $G$ is a \emph{copy} or an \emph{instance} of $H$ if $H'$ is isomorphic to $H$. An isomorphism-preserving mapping from $H$ to a copy of $H$ in $G$ is called an \emph{embedding} of $H$ in $G$.

\subparagraph{Edge Cover and Graph Decomposition}
We use the following definition of the fractional edge cover of a graph and a decomposition result based on it by Assadi, Kapralov and Khanna~\cite{assadi2018simple}.

\begin{definition}[Fractional Edge-Cover Number]
	A fractional edge-cover of $H(V_H,E_H)$ is a mapping $\psi: E_H \rightarrow [0,1]$ such that for each vertex $v\in V_H$, $\sum_{e\in E_H, v\in e} \psi(e)\geq 1$. The fractional edge-cover number $\rho(H)$ of $H$ is the minimum value of $\sum_{e\in E_H}\psi(e)$ among all fractional edge-covers $\psi$.
\end{definition}

Let $C_k$ denote the cycle of length $k$. Let $S_k$ denote a star with $k$ petals, i.e., $S_k = (\{u, v_1, \ldots, v_k\}, \cup_{i \in [k]} \{u, v_k\})$. Let $K_k$ denote a clique on $k$ vertices. It is known that $\rho(C_{2k+1})=k+1/2$, $\rho(S_k)=k$ and $\rho(K_k)=k/2$. 

\begin{lemma}[\cite{assadi2018simple}]
	\label{decomposition-lemma}
	Any subgraph $H$ can be decomposed into a collection of vertex-disjoint odd cycles $\overline{C_1},\ldots,\overline{C_o}$ and star graphs $\overline{S_1},\ldots,\overline{S_s}$ such that 
	$$\rho(H)=\sum_{i=1}^o\rho(\overline{C_i})+\sum_{j=1}^s\rho(\overline{S_j}).$$
\end{lemma}

By a result of Atserias, Grohe and Marx \cite{atserias2008size}, the number of instances of $H$ in a graph $G$ with $m$ edges is $O(m^{\rho(H)})$.

\section{Sampling an Arbitrary Subgraph $H$}
\label{sec:upperbound}

In this section, we present sampling algorithms for odd cycles and stars and show how to combine them to obtain a sampling algorithm for arbitrary subgraphs.

\subsection{Sampling an Odd-Length Cycle}

\begin{algorithm}
	\caption{Sampling a wedge}
	\label{wedge-sample}
	\begin{algorithmic}[1]
		\Procedure{\textsc{SampleWedge}}{$G,u,v$}
		\If{$d_u \leq \sqrt{2 m}$} \label{wedge-degree}
		\State sample a number $i \in \{ 1, \ldots \sqrt{2 m} \}$ uniformly at random
		\If{$i > d_u$}
		\State \Return \fail
		\EndIf
		\State $w$ $\leftarrow$ $i^{th}$ neighbor of $u$
		\Else
		\State  %
		sample a vertex $w$ with probability proportional to its degree
		\State sample a number $t \in [0,1]$ uniformly at random
		\If{$t > \sqrt{2m} / d_w$}
		\State \Return \fail
		\EndIf
		\EndIf
		\State \Return $w$
		\EndProcedure
	\end{algorithmic}
\end{algorithm}

\begin{algorithm}
	\caption{Sampling a cycle of length $2k+1$}
	\label{odd-cycle-sample}
	\begin{algorithmic}[1]
		\Procedure{\textsc{SampleOddCycle}}{$G,2k+1$}  
		\State sample $k$ directed edges $\diedge{u_1,v_1}, \ldots, \diedge{u_k,v_k}$ u.a.r. and i.i.d. \label{odd-cycle-loop}
		\State $w$ $\leftarrow$ \textsc{SampleWedge}($G,u_1,v_{k}$)
		\If{$u_1 \prec w \prec v_1$, and $\forall i > 1 : u_1 \prec u_i, v_i$} \label{path-check}
		\State \Return $\{(u_1,v_1),\ldots,(u_k,v_k)\}\cup \{ (v_{k},w), (w,u_1) \}$
		\EndIf
		\State \Return\fail
		\EndProcedure
	\end{algorithmic}
\end{algorithm}

We describe our algorithm \textsc{SampleOddCycle} for sampling a uniformly random odd-length $k$-cycle. For any instance $C$ of $C_{2k+1}$ in the input graph, our goal is to guarantee that it will be sampled with probability $\frac{1}{m^{k+1/2}}$. Let $e_1, \ldots, e_{2k+1}$ be a sequence of edges that represents a cycle of length $2k+1$. While we can use edge sampling to sample every second edge of the first $2k$ edges sequentially, i.e., $e_1, e_3, \ldots, e_{2k-1}$, and query the edges inbetween, i.e., $e_2, \ldots, e_{2k-2}$, by vertex pair queries, we use a different strategy to sample $e_{2k}$ and $e_{2k+1}$. Let $\{u,v\} = e_1$. If $u$ has low degree, i.e., $d_u \leq \sqrt{2m}$, we can afford to sample each neighbor of $u$ with probability $1 / \sqrt{2m}$ and fail if no neighbor is sampled. If $d_u > \sqrt{2m}$, we reduce the number of candidate neighbors to $\sqrt{2m}$ and sample from the remaining candidates uniformly at random. To this end, we define a unique embedding of $C$ such that $u$ is the smallest vertex according to the order ``$\prec$''. By \cref{lemma: tu-upper-bound}, it follows that the number of candidate neighbors is at most $\sqrt{2m}$. Another reason to accept exactly one embedding of $C$ is that there exists a linear number of automorphisms for every cycle. If we would accept every embedding, bounding the probability that every instance of $C_{2k+1}$ is sampled \emph{exactly} uniformly is hard as some instance might be sampled less likely because, e.g., its edges participate in many overlapping cycles.

In particular, we sample $k$ directed edges $(u_1, v_1), \ldots, (u_k, v_k)$ independently and uniformly at random and call \textsc{SampleWedge} on $u_1, v_k$. Then, we require that $u_1$ is the (unique) smallest vertex according to the order ``$\prec$''  among all $u_i, v_i, i \geq 1$ and $w$. This leaves only two orientations of the cycle that are distinguished by $w \prec v_1$ and $v_1 \prec w$. We (arbitrarily) choose $w \prec v_1$. If any of these requirements is not met, we have not sampled the uniquely defined embedding we are looking for, and the algorithm fails.

\begin{lemma}
	\label{odd-cycle-sampler-lemma}
	For any instance of an odd cycle $C_{2k+1}$ in $G$, the probability that it will be returned by \textsc{SampleOddCycle}($G,2k+1$) is $\frac{1}{(2m)^{k+1/2}}$.
\end{lemma}
\begin{proof}
	Let $\mathcal{C}_{2k+1}$ be any instance of a cycle of odd length $2k+1$ in $G$. Let $x_0$ be the smallest vertex on $\mathcal{C}_{2k+1}$ according to the total order ``$\prec$''. Let $x_1, x_{2k}$ be the two neighbors of $x_0$ on $\mathcal{C}_{2k+1}$ such that $x_1\prec x_{2k}$. Then, we let $x_i$ denote the vertices on $\mathcal{C}_{2k+1}$ such that $(x_i, x_{i+1}) \in E(\mathcal{C}_{2k+1})$ for $0\leq i\leq 2k-1$ and $(x_{2k},x_0)\in E(\mathcal{C}_{2k+1})$. Note that for any $\mathcal{C}_{2k+1}$, there is a \emph{unique} way of mapping its vertices to $x_i$, for $0\leq i\leq 2k$. 
	Thus, $\textsc{SampleOddCycle}$ returns $\mathcal{C}_{2k+1}$ if and only if 
	\begin{enumerate}
		\item $u_1 = x_0$ and $v_1 = x_{2k}$; \label{odd-mapping-a}
		\item $u_i=x_{2k-2i+3}$ and $v_i=x_{2k-2i+2}$ for $2\leq i\leq k$; \label{odd-mapping-b}
		\item \textsc{SampleWedge}($G,u_{1},v_k$) returns $x_1$. \label{odd-mapping-c}
	\end{enumerate}
	
	Event \ref{odd-mapping-a} occurs with probability $1/(2m)$, and event \ref{odd-mapping-b} occurs with probability $1/(2m)^{k-1}$, as each directed edge is sampled with probability $1/(2m)$. 
	
	Now we bound the probability of event \ref{odd-mapping-c}. In the call to \textsc{SampleWedge}, let $u := u_1$ and $v := v_k$, which satisfies that $u \prec v$. We first note that if $d_u < \sqrt{2m}$ in \textsc{Sample\-Wedge}($G,u_{1},v_k$), then the vertex $x_1$ will be sampled with probability $1 / \sqrt{2m}$. Now we consider the case that $d_u \geq \sqrt{2m}$. It follows that $d_{x_1} \geq d_u > \sqrt{2m}$. In this case, we will sample $x_1$ with probability $\frac{d_{x_1}}{2m} \cdot t = \frac{d_{x_1}}{2m} \cdot \frac{\sqrt{2m}}{d_{x_1}} = \frac{1}{\sqrt{2m}}$. Thus in both cases, the probability that event \ref{odd-mapping-c} occurs is $\frac{1}{\sqrt{2m}}$.
	Therefore, the probability that \textsc{SampleOddCycle} returns $\mathcal{C}_{2k+1}$ is $\frac{1}{\sqrt{2m}} \cdot \frac{1}{2m} \cdot (\frac{1}{2m})^{k-1}=\frac{1}{(2m)^{k+1/2}}$. \qedhere
\end{proof}

\subsection{Sampling a Star}

Similarly to odd cycles, we observe that every $k$-star admits an exponential number of automorphisms. Therefore, we enforce a unique embedding of every instance of a $k$-star in our sampling algorithm \textsc{SampleStar}. Let $e_1, \ldots, e_k$ be the petals of an instance of a $k$-star. We sample $e_1, \ldots, e_k$ sequentially. If these edges form a star, we output it only if the leaves where sampled in ascending order with respect to ``$\prec$''.

\begin{algorithm}[H]
	\caption{Sampling a star with $k$ petals}
	\label{star-sample}
	\begin{algorithmic}[1]
		\Procedure{SampleStar}{$G,k$}
		\State Sequentially sample $k$ directed edges $ \{\diedge{u_1, v_1}, \ldots, \diedge{u_{k}, v_{k}}\}$ u.a.r. and i.i.d.
		\If{$u_1 = u_2 = \ldots = u_k$ and $v_1 \prec v_2 \prec \ldots \prec v_k$
		}
		\State \Return $(u_1, v_1, \ldots, v_k)$
		\EndIf
		\State \Return \fail
		\EndProcedure
	\end{algorithmic}
\end{algorithm}

\begin{lemma}
	\label{star-sampler-lemma}
	For any instance of a $k$-star $S_k$ in $G$, the probability that it will be returned by the algorithm \textsc{SampleStar}($G,k$) is $\frac{1}{(2m)^k}$.
\end{lemma}
\begin{proof}
	Consider any instance of $S_k$ with root $x$ and petals $y_1, \ldots, y_k$ such that $y_1 \prec \ldots y_k$.
	Note that it will be returned by \textsc{SampleStar} if and only if all the directed edges $\diedge{x,y_1},\ldots,\diedge{x,y_k}$ are sequentially sampled, which occurs with probability $1 /(2m)^k$.
\end{proof}

\subsection{Sampling $H$}
Let $H$ be a subgraph. It can be decomposed
into collections of $o$ odd cycles $\overline{C_i}$ and $s$ stars $\overline{S_j}$ as given in \cref{decomposition-lemma}. We say that $H$ has a (decomposition) \emph{type}  $\overline{T}=\{\overline{C_1},\ldots,\overline{C_o},\overline{S_1},\ldots,\overline{S_s}\}$. 
\begin{definition}
	Given a graph $G$, for each potential \emph{instance} $\mathcal{H}$ of $H$, we say that $\mathcal{H}$ can be decomposed into \emph{configurations} $\mathcal{T}=\{\mathcal{C}_1,\ldots,\mathcal{C}_o,\mathcal{S}_1,\ldots,\mathcal{S}_s\}$ with respect to type $\overline{T}=\{\overline{C_1},\ldots,\overline{C_o},\overline{S_1},\ldots,\overline{S_s}\}$, if
	\begin{enumerate}
		\item $\mathcal{C}_i \cong \overline{C_i}$ for any $1\leq i\leq o$, and $\mathcal{S}_j \cong \overline{S}_j$, for any $1\leq i\leq s$
		\item all the remaining edges of $H$ between vertices specified in $\mathcal{T}$ all are present in $G$.
	\end{enumerate}
	We let $f_{\overline{T}}(H)$ denote the number of all possible configurations $\mathcal{T}$ into which $H$ can be decomposed with respect to $\overline{T}$.
\end{definition}
\begin{algorithm}[H]
	\caption{Sampling a copy of subgraph $H$}
	\label{subgraph-sample}
	\begin{algorithmic}[1]  
		\Procedure{SampleSubgraph}{$G, H$}
		\State{Let $\overline{T}=\{\overline{C_1},\ldots,\overline{C_o},\overline{S_1},\ldots,\overline{S_s}\}$ denote a (decomposition) type  of $H$. }
		\ForAll{$i=1\ldots o$}
		\If{ \textsc{SampleOddCycle($G, \lvert E(\overline{C}_i) \rvert$)} returns a cycle $\mathcal{C}$}
		\State $\mathcal{C}_i \gets \mathcal{C}$\label{alg:cycle_H}
		\Else
		\State \Return \fail
		\EndIf
		\EndFor
		\ForAll{$j =1\ldots s$}
		\If{\textsc{SampleStar($G,\lvert V(\overline{S}_j) \rvert - 1$)} returns a star $\mathcal{S}$}
		\State $\mathcal{S}_j \gets \mathcal{S}$\label{alg:star_H} 
		\Else
		\State \Return \fail
		\EndIf 
		\EndFor
		
		\State Query all edges $(\bigcup_{i \in [o]} V(\mathcal{C}_i) \cup \bigcup_{j \in [s]} V(\mathcal{S}_j))^2$
		\If{$S := (\mathcal{C}_1, \ldots, \mathcal{C}_o, \mathcal{S}_1, \ldots, \mathcal{S}_s)$ forms a copy of $H$}
		\State flip a coin and with probability $\frac{1}{f_{\overline{T}}(H)}$: \Return $S$ \label{alg:occurrence_H}
		\EndIf
		\State \Return \fail
		\EndProcedure
	\end{algorithmic}
\end{algorithm}

\begin{lemma}
	\label{subgraph-sampler-lemma}
	For any instance of a subgraph $H$ in $G$, the probability that it will be returned by the algorithm \textsc{SampleSubgraph}($G, H$) is $\frac{1}{(2m)^{\rho(H)}}$.
\end{lemma}
\begin{proof} 
	
	For any instance $\mathcal{H}$ of $H$ in $G$, and any configuration $\mathcal{T}=\{\mathcal{C}_1,\ldots,\mathcal{C}_O,\mathcal{S}_1,\ldots,\mathcal{S}_s\}$ of $\mathcal{H}$ with respect to $\overline{T}$, $\mathcal{H}$ will be returned by \textsc{SampleSubgraph}($G,H$) if and only if
	\begin{enumerate}
		\item $\mathcal{C}_i$ is returned in \cref{alg:cycle_H} for each $1\leq i\leq o$, and $\mathcal{S}_j$ is returned in \cref{alg:star_H} for any $1\leq j\leq s$; 
		\item the configuration is returned with probability $\frac{1}{f_{\overline{T}}(H)}$ in \cref{alg:occurrence_H}.
	\end{enumerate}
	By Lemma \ref{odd-cycle-sampler-lemma}, each $\mathcal{C}_i$ will be returned with probability $\frac{1}{(2m)^{|E(\overline{C_i})|/2}}=\frac{1}{(2m)^{\rho(\overline{C_i})}}$. By Lemma \ref{star-sampler-lemma} each $\mathcal{S}_j$ will be returned with probability $\frac{1}{(2m)^{|V(\overline{S_j})|-1}}=\frac{1}{(2m)^{\rho(\overline{S_j})}}$. Thus, $\mathcal{T}$ will be returned with probability 
	\begin{equation*}
		\prod_{i=1}^o \frac{1}{(2m)^{\rho(\overline{C_i})}}\cdot \prod_{j=1}^s \frac{1}{(2m)^{\rho(\overline{S_j})}} \cdot \frac{1}{f_{\overline{T}}(H)} = \frac{1}{(2m)^{\rho(H)}}\cdot \frac{1}{f_{\overline{T}}(H)}.
	\end{equation*}
	
	Finally, since there are $f_{\overline{T}}(H)$ configurations of $\mathcal{H}$ with respect to $\overline{T}$, the instance will be returned with probability $f_{\overline{T}}(H)\cdot \frac{1}{(2m)^{\rho(H)}} \cdot \frac{1}{f_{\overline{T}}(H)} = \frac{1}{(2m)^{\rho(H)}}$. \qedhere
\end{proof}

\subsection{The Final Sampler}
Let $X_H$ be an estimate of $\# H$. Such an estimate can be obtained by, e.g., the subgraph counting algorithm of Assadi, Kapralov and Khanna~\cite{assadi2018simple} in expected time $\tilde{O}(m^{\rho(H)}/\# H)$. We show that by sufficiently many calls to \textsc{SampleSubgraph}, we can obtain a uniformly random sample of an instance of $H$ with constant probability.

\begin{algorithm}[H]
	\caption{Sampling a copy of subgraph $H$ uniformly at random}
	\label{subgraph-sample-uniformly}
	\begin{algorithmic}[1]  
		\Procedure{SampleSubgraphUniformly}{$G, H, X_H$}
		\ForAll{$j=1,\ldots, q=10 \cdot {(2m)}^{\rho(H)}/X_H$}
		\State Invoke \textsc{SampleSubgraph}($G,H$)
		\If{a subgraph $H$ is returned}
		\Return $H$
		\EndIf
		\EndFor
		\State \Return \fail
		\EndProcedure
	\end{algorithmic}
\end{algorithm}

\begin{lemma}
	\label{correctness}
	If $\# H \leq X_H\leq 2 \# H$, then Algorithm \textsc{SampleSubgraphUniformly}$(G,{H}, X_H)$ returns a copy $H$ with probability at least $2/3$. The distribution induced by the algorithm is (exactly) uniform over the set of all instances of $H$ in $G$. 
\end{lemma}
\begin{proof}
	Since $\# H \leq X_H\leq 2 \# H$, the probability that no instance of $H$ is returned in $q=10 \cdot {(2m)}^{\rho(H)}/X_H$ invocations is at most 
	\[
	\left(1-\frac{\# H}{(2m)^{\rho(H)}}\right)^{q} \leq e^{-\frac{\# H}{(2m)^{\rho(H)}} \cdot q} < \frac13
	\]
	by \cref{subgraph-sampler-lemma}. Let $\mathcal{H}$ be an instance of $H$. By Lemma \ref{subgraph-sampler-lemma}, the probability that \textsc{SampleSubgraph}($H$) returns $\mathcal{H}$ is $\frac{1}{(2m)^{\rho(H)}}$. Thus, the probability that  \textsc{SampleSubgraphUniformly}$(G,{H})$ successfully output an instance of $H$ is 
	\begin{align*}
		\frac{\# H}{(2m)^{\rho(H)}}.
	\end{align*}
	
	Conditioned on the event that \textsc{SampleSubgraphUniformly}$(G,{H})$ succeeds, the probability that any specific instance $\mathcal{H}$ will be returned is
	\begin{align*}
		p_\mathcal{H} = \frac{\frac{1}{(2m)^{\rho(H)}}}{\frac{\# H}{(2m)^{\rho(H)}}} = \frac{1}{\# H}.
	\end{align*}
	
	That is, with probability at least $\frac23$, an instance $\mathcal{H}$ is sampled from the uniform distribution over all the instances of $H$ in $G$. \qedhere
	
\end{proof}   

Finally, we prove the expected query and time complexity of \textsc{SampleSubgraphUniformly}.

\begin{lemma}
	\label{complexity}
	The expected query and time complexity of \textsc{SampleSubgraph}-\\ \textsc{Uniformly}$(G,H,X_H)$ is $O(m^{\rho(H)} / X_H)$.
\end{lemma}

\begin{proof}
	The query complexity of \textsc{SampleOddCycle}$(G,2k+1)$ is $O(1)$, as the query complexity of \textsc{SampleWedge}$(G, u_1, v_k)$ is $O(1)$. The query complexity of \textsc{SampleStar}$(G,k)$ is bounded by $k \in O(1)$. It follows that the query complexity of \textsc{SampleSubgraph}$(G, H)$ is at most $(o + s + \lvert H \rvert^2) \cdot O(1) \subseteq \lvert H \rvert \cdot O(1)$. The query complexity of \textsc{SampleSubgraphUniformly}$(G, H)$ is $O({(2m)}^{\rho(H)}/X_H \cdot \lvert H \rvert^2)=\tilde{O}({(2m)}^{\rho(H)}/X_H)$. To bound the running time, we observe that every loop in our algorithm issues at least one query, and we only perform isomorphism checks on subgraphs of constant size. Thus the running time is still $\tilde{O}({(2m)}^{\rho(H)}/X_H)$.
\end{proof}

The proof of \cref{thm:main} follows almost directly from \cref{correctness,complexity}.

\begin{proof}[Proof of \cref{thm:main}]
	For the case that $m \geq m^{\rho(H)} / \# H$, the claim follows from \cref{correctness,complexity}. If $m < m^{\rho(H)} / \# H$, we can query the whole graph, which requires $O(m)$ degree and neighbor queries, store the graph and answer the queries of the algorithm from this internal memory.
\end{proof}

Finally, we remark that our algorithm assumes the knowledge of the number of edges $m$. This assumption can be lifted with a slight cost on the query complexity and approximation of our algorithm: for any constant $\varepsilon>0$, one can first find a $(1+\varepsilon)$-approximation of $m$ by making $\tilde{O}(\sqrt{n})$ neighbor queries \cite{edge_counting_0}. We can then guarantee that the returned copy of $H$ is sampled with probability $(1\pm \Theta(\varepsilon))\cdot \frac{1}{\# H}$. That is, we simply replace $m$ by a $(1+\varepsilon)$-approximation of $m$ in \cref{wedge-sample,subgraph-sample-uniformly} and then the performance guarantee of the resulting algorithm directly follows from the previous analysis.

\section{Proof of Theorem \ref{thm:lowerbound}}\label{sec:lowerbound}
In this section, we give the proof of Theorem~\ref{thm:lowerbound}, which follows by adapting the proofs for the lower bounds on the query complexity for approximate counting subgraphs given by Eden and Rosenbaum \cite{total-lower-bound}.

\begin{theorem}[see Theorems 4.7 and B.1 in \cite{total-lower-bound}]
	\label{thm:lower-bound-families}
	For any choices of $n,m,r,c_r > 0$, there exist families of graphs with $n$ vertices and $m$ edges, $\mathcal{F}_0$ and $\mathcal{F}_1$, such that 
	\begin{itemize}
		\item all graphs in $\mathcal{F}_0$ are $K_r$-free,
		\item all graphs in $\mathcal{F}_1$ contain at least $c_r$ copies of $K_r$, 
		\item and any algorithm in the augmented general graph model that distinguishes a graph $G \in \mathcal{F}_0$ from $G \in \mathcal{F}_1$ with probability $\Omega(1)$ requires $\Omega(\min\{m, m^{r/2} / c_r(cr)^r\})$ queries for some constant $c>0$.
	\end{itemize}
\end{theorem}

Now we prove our Theorem \ref{thm:lowerbound}.
\begin{proof}[Proof of Theorem \ref{thm:lowerbound}]
	Let $\mathcal{A}$ be an algorithm that for any graph $G=(V,E)$ on $n$ vertices and $m$ edges returns an arbitrary $r$-clique $K_r$, if one exists; and each $K_r$ is sampled according to $\mathcal{D}$, using $f(m, r, \# K_r) \in o(\min\{m, \frac{m^{r/2}}{\# K_r\cdot (cr)^r}\})$ neighbor, degree, pair and edge sampling queries.
	
	Let  $n, m, c_r > 0$ and let $\mathcal{F}_0, \mathcal{F}_1$ be the families from \cref{thm:lower-bound-families}. Consider the following algorithm $\mathcal{A}'$: run $\mathcal{A}$ on a graph from $\mathcal{F}_0 \cup \mathcal{F}_1$ and terminate $\mathcal{A}$ if it did not produce a $K_r$ after $f(m, r, c_r)$ queries. If it output a clique, $\mathcal{A}'$ claims that $G \in \mathcal{F}_1$, otherwise it claims that $G \in \mathcal{F}_0$. By the assumption, $\mathcal{A}$ returns a clique after at most $f(m, r, c_r)$ queries with probability $\Omega(1)$ if $G \in \mathcal{F}_1$ because then $G$ contains at least $c_r$ copies of $K_r$ and the probability mass of $\mathcal{D}$ on the set of all copies of $K_r$ is $\Omega(1)$. Otherwise, $G \in \mathcal{F}_0$, which implies that $G$ contains no triangle. Therefore, $\mathcal{A}$ cannot output a triangle from $G$.
	
	It follows that $\mathcal{A}'$ can distinguish $\mathcal{F}_0$ and $\mathcal{F}_1$, which is a contradiction to \cref{thm:lower-bound-families}.
\end{proof}
\section*{Acknowledgments}
We would like to thank the anonymous reviewers for their detailed comments. In particular, we would like to thank an anonymous reviewer for their suggestion to improve the presentation of the proof of \cref{thm:lowerbound} and their comment on applications, which we included as future work.

\bibliographystyle{alpha}
\bibliography{mybibliography}

\newcommand{\etalchar}[1]{$^{#1}$}
\begin{thebibliography}{MSOI{\etalchar{+}}02}

\bibitem[ABG{\etalchar{+}}17]{AliSub17}
Maryam Aliakbarpour, Amartya~Shankha Biswas, Themis Gouleakis, John Peebles,
  Ronitt Rubinfeld, and Anak Yodpinyanee.
\newblock Sublinear-{{Time Algorithms}} for {{Counting Star Subgraphs}} via
  {{Edge Sampling}}.
\newblock {\em Algorithmica}, 2017.

\bibitem[ADWR17]{ahmed2017sampling}
Nesreen~K Ahmed, Nick Duffield, Theodore~L Willke, and Ryan~A Rossi.
\newblock On sampling from massive graph streams.
\newblock {\em Proceedings of the VLDB Endowment}, 10(11), 2017.

\bibitem[AGM08]{atserias2008size}
Albert Atserias, Martin Grohe, and D{\'a}niel Marx.
\newblock Size bounds and query plans for relational joins.
\newblock In {\em Proceedings of the 49th Annual IEEE Symposium on Foundations
  of Computer Science (FOCS)}, 2008.

\bibitem[AKK18]{assadi2018simple}
Sepehr Assadi, Michael Kapralov, and Sanjeev Khanna.
\newblock {A Simple Sublinear-Time Algorithm for Counting Arbitrary Subgraphs
  via Edge Sampling}.
\newblock In {\em Proceedings of the 10th Innovations in Theoretical Computer
  Science Conference (ITCS)}, volume 124. Schloss Dagstuhl--Leibniz-Zentrum
  fuer Informatik, 2018.

\bibitem[BHOP18]{Ben-Hamou2018}
Anna Ben-Hamou, Roberto~I Oliveira, and Yuval Peres.
\newblock Estimating graph parameters via random walks with restarts.
\newblock In {\em Proceedings of the Twenty-Ninth Annual ACM-SIAM Symposium on
  Discrete Algorithms}, 2018.

\bibitem[CDK{\etalchar{+}}16]{Chiericetti2016}
Flavio Chiericetti, Anirban Dasgupta, Ravi Kumar, Silvio Lattanzi, and
  Tam\'{a}s Sarl\'{o}s.
\newblock On sampling nodes in a network.
\newblock In {\em Proceedings of the 25th International Conference on World
  Wide Web}, 2016.

\bibitem[ELRS17]{triangle_counting_1}
Talya Eden, Amit Levi, Dana Ron, and C~Seshadhri.
\newblock Approximately counting triangles in sublinear time.
\newblock {\em SIAM Journal on Computing}, 46(5), 2017.

\bibitem[ER17]{edge-sampling}
Talya Eden and Will Rosenbaum.
\newblock On sampling edges almost uniformly.
\newblock {\em arXiv preprint arXiv:1706.09748}, 2017.

\bibitem[ER18a]{total-lower-bound}
Talya Eden and Will Rosenbaum.
\newblock Lower {{Bounds}} for {{Approximating Graph Parameters}} via
  {{Communication Complexity}}.
\newblock In {\em Approximation, {{Randomization}}, and {{Combinatorial
  Optimization}}. {{Algorithms}} and {{Techniques}} ({{APPROX}}/{{RANDOM}})},
  volume 116. {Schloss Dagstuhl\textendash{}Leibniz-Zentrum fuer Informatik},
  2018.

\bibitem[ER18b]{edge_sampling_1}
Talya Eden and Will Rosenbaum.
\newblock {On Sampling Edges Almost Uniformly}.
\newblock In {\em 1st Symposium on Simplicity in Algorithms (SOSA)}, volume~61.
  Schloss Dagstuhl--Leibniz-Zentrum fuer Informatik, 2018.

\bibitem[ERR19]{eden_et_al:LIPIcs:2019:10628}
Talya Eden, Dana Ron, and Will Rosenbaum.
\newblock {The Arboricity Captures the Complexity of Sampling Edges}.
\newblock In {\em 46th International Colloquium on Automata, Languages, and
  Programming (ICALP)}, volume 132. Schloss Dagstuhl--Leibniz-Zentrum fuer
  Informatik, 2019.

\bibitem[ERS18]{eden2018approximating}
Talya Eden, Dana Ron, and C~Seshadhri.
\newblock On approximating the number of k-cliques in sublinear time.
\newblock In {\em Proceedings of the 50th Annual ACM SIGACT Symposium on Theory
  of Computing}, 2018.

\bibitem[ERS20]{ERS20}
Talya Eden, Dana Ron, and C~Seshadhri.
\newblock Faster sublinear approximation of the number of $k$-cliques in
  low-arboricity graphs.
\newblock In {\em Proceedings of the Fourteenth Annual ACM-SIAM Symposium on
  Discrete Algorithms}, pages 1467--1478. SIAM, 2020.

\bibitem[Fei06]{edge_counting_app_1}
Uriel Feige.
\newblock On sums of independent random variables with unbounded variance and
  estimating the average degree in a graph.
\newblock {\em SIAM Journal on Computing}, 35(4), 2006.

\bibitem[FSY17]{feng2017can}
Weiming Feng, Yuxin Sun, and Yitong Yin.
\newblock What can be sampled locally?
\newblock In {\em Proceedings of the ACM Symposium on Principles of Distributed
  Computing (PODC)}, 2017.

\bibitem[GR08]{edge_counting_0}
Oded Goldreich and Dana Ron.
\newblock Approximating average parameters of graphs.
\newblock {\em Random Structures \& Algorithms}, 32(4), 2008.

\bibitem[HL13]{hu2013survey}
Pili Hu and Wing~Cheong Lau.
\newblock A survey and taxonomy of graph sampling.
\newblock {\em arXiv preprint arXiv:1308.5865}, 2013.

\bibitem[JVV86]{jerrum1986random}
Mark~R Jerrum, Leslie~G Valiant, and Vijay~V Vazirani.
\newblock Random generation of combinatorial structures from a uniform
  distribution.
\newblock {\em Theoretical Computer Science}, 43, 1986.

\bibitem[KKR04]{edge_sampling_begin}
Tali Kaufman, Michael Krivelevich, and Dana Ron.
\newblock Tight bounds for testing bipartiteness in general graphs.
\newblock {\em SIAM Journal on Computing}, 33(6), 2004.

\bibitem[LF06]{leskovec2006sampling}
Jure Leskovec and Christos Faloutsos.
\newblock Sampling from large graphs.
\newblock In {\em Proceedings of the 12th ACM SIGKDD international conference
  on Knowledge discovery and data mining}, 2006.

\bibitem[MSOI{\etalchar{+}}02]{triangle_counting_app_3}
R.~Milo, S.~Shen-Orr, S.~Itzkovitz, N.~Kashtan, D.~Chklovskii, and U.~Alon.
\newblock Network motifs: Simple building blocks of complex networks.
\newblock {\em Science}, 298(5594), 2002.

\bibitem[SJ89]{sinclair1989approximate}
Alistair Sinclair and Mark Jerrum.
\newblock Approximate counting, uniform generation and rapidly mixing markov
  chains.
\newblock {\em Information and Computation}, 82(1), 1989.

\end{thebibliography}
\end{document}